\documentclass[12pt,a4paper,onecolumn]{IEEEtran}
\title{\Large{Performance Characterization and Transmission Schemes for Instantly Decodable Network Coding in Wireless Broadcast}}
\author{Mingchao Yu, Parastoo Sadeghi, and Neda Aboutorab\\ \small{Research School of Engineering, The Australian National University, Canberra, Australia}\\ \texttt{Email: \{ming.yu,parastoo.sadeghi,neda.aboutorab\}@anu.edu.au}}
\usepackage{amsopn,amsmath,amssymb,amsfonts,bm,amsthm}
\usepackage{graphicx,subfigure}
\usepackage{enumerate}
\usepackage{mdwlist}
\usepackage{cite}
\usepackage{etex}
\usepackage{algorithm,algorithmic}
\usepackage{tikz,pgfplots}
\usepackage{array,multirow}
\usepackage[margin=16mm,top=20mm,bottom=16mm]{geometry}
\usetikzlibrary{plotmarks}
\pgfplotsset{compat=1.5}

\newcommand{\set}[1]{\mathcal{#1}}

\def\e{\boldsymbol e}
\def\E{\set{E}}

\def\F{\mathbb{F}}
\def\G{\set{G}}

\def\M{\set{M}}

\def\P{\set{P}}
\def\p{\mathbf p}
\def\v{\boldsymbol v}
\def\V{\set{V}}
\def\X{\boldsymbol X}
\def\S{\set{S}}
\def\cG{\overline{\set{G}}}

\def\mA{\boldsymbol A}

\def\T{\mathcal T}
\def\R{\mathcal R}

\def\W{\set{W}}

\newcommand{\figref}[1]{Fig.\,\ref{#1}}

\newtheorem{Theorem}{\textbf{Theorem}}
\newtheorem{Corollary}{Corollary}
\newtheorem{Definition}{\textbf{Definition}}

\newtheorem{Remark}{\textbf{Remark}}

\newtheorem{Example}{\textbf{Example}}
  {\proof}{\proofend}

\newtheorem{property}{Property}
\newtheorem{defn}{Definition}[Definition]
\begin{document}
\maketitle
\vspace{-5em}
\begin{abstract}
We consider broadcasting a block of packets to multiple wireless receivers under random packet erasures using instantly decodable network coding (IDNC).
The sender first broadcasts each packet uncoded once, then generates coded packets according to receivers' feedback about their missing packets.
We focus on strict IDNC (S-IDNC), where each coded packet includes at most one missing packet of every receiver. But we will also compare it with general IDNC (G-IDNC), where this condition is relaxed. We characterize two fundamental performance limits of S-IDNC: 1) the number of transmissions to complete the broadcast, and 2) the average delay for a receiver to decode a packet. We derive a closed-form expression for the expected minimum number of transmissions in terms of the number of packets and receivers and the erasure probability. We prove that it is NP-hard to minimize the decoding delay of S-IDNC. We also derive achievable upper bounds on the above two performance limits. We show that G-IDNC can outperform S-IDNC
without packet erasures, but not necessarily with packet erasures. Next, we design optimal and heuristic S-IDNC transmission schemes and coding algorithms with full/intermittent receiver feedback. We present simulation results to corroborate the developed theory and compare with existing schemes.
\end{abstract}

\emph{Index terms--} Wireless broadcast, network coding, throughput, decoding delay, instantly decode.
\section{Introduction}

The broadcast nature of wireless medium allows one sender to simultaneously serve multiple receivers who are interested in the same data. We consider a block-based wireless broadcast system where a sender wishes to deliver a block of data packets to a set of receivers, with the channels between the sender and the receivers are subject to independent random packet erasures. A traditional approach is to send the data packets unaltered under a receiver feedback mechanism, such as Automatic-Repeat-reQuest (ARQ) \cite{medard:arq}. This approach, though simple, is inefficient in terms of throughput, as the transmitted packets are non-innovative to the receivers who have already received them.

The advent of network coding \cite{Yeung_flow} starts a new era for high throughput network coded wireless communications \cite{katti:etal:2008,fragouli:widmer:boudec:2008,keller:drinea:fragouli:2008,
eryilmaz:ozdaglar:medard:ahmed:2008,tran2008joint,nguyen:tran:nguyen:bose:2009,lucani2009broadcasting,
lucani:TDD:fieldsize,heide_systematic_RLNC,sameh:valaee:globecom:2010,Rozner_Heuristic_clique,sadeghi:adaptive_broadcast_2009,
sadeghi:shams:traskov:2010,parastoo:yu:neda:isita,yu:parastoo:neda:2014}.
By linearly adding all data packets together with randomly chosen coefficients from a sufficiently large finite field, random linear network coding (RLNC) can almost surely achieve the optimal block completion time in block-based wireless broadcast \cite{ho:medard:koetter:karger:effros:2006,lucani2009broadcasting,lucani:TDD:fieldsize,heide_systematic_RLNC}, which is defined as the number of transmissions it takes to complete the broadcast, and is a fundamental measure of throughput due to their inverse relation under a fixed block size. Compared to other optimal codes such as Fountain codes \cite{luby2002lt,shokrollahi2006raptor}, RLNC is preferred due to its ease of implementation at the sender and extension to more complex networks and traffic settings.

However, with RLNC, data packets are block-decoded by solving a set of linear equations, which only takes place after a sufficient number of coded packets have been received. RLNC thus may suffer from heavy computational load \cite{heide_systematic_RLNC} and packet decoding delay \cite{yu:parastoo:neda:2014}, which is measured by the average time it takes for a receiver to decode a data packet. The first issue can, for example, hinder the application of RLNC for mobile receivers with limited computational capability \cite{heide_systematic_RLNC}. Meanwhile, a large packet decoding delay can be unacceptable for delay-sensitive applications such as video streaming \cite{li:idnc_video:2011,yang:sagduyu:li:Zhang:2012}.

To mitigate these issues, instantly decodable network coding (IDNC) techniques \cite{katti:etal:2008,Rozner_Heuristic_clique,sadeghi:adaptive_broadcast_2009,sameh:valaee:globecom:2010,sadeghi:shams:traskov:2010} have been introduced. With IDNC, the sender first broadcasts the data packets uncoded once. It then makes online coding decisions based on receivers' feedback about their packet reception state, under the restriction that coding/decoding is over binary field. A simple packet reception state is demonstrated in Table \ref{tab:simple_sfm}. There are two data packets, $\p_1$ and $\p_2$, and three receivers, $R_{1}$ to $R_3$, where each has a subset of $\{\p_1,\p_2\}$ and wants the rest. Consider a coded packet of $\X=\p_1\oplus\p_2$, where $\oplus$ denotes binary XOR. It has three different effects on different receivers: 1) it is instantly decodable to $R_1$, because $R_1$ can decode $\p_2$ by performing $\X\oplus\p_1$; 2) it is non-instantly decodable to $R_2$, because $R_2$ has neither $\p_1$ nor $\p_2$; and 3) it is non-innovative to $R_3$, because $R_3$ already has both $\p_1$ and $\p_2$.

There are two main types of IDNC techniques. The first one, called strict IDNC (S-IDNC) \cite{Rozner_Heuristic_clique,sadeghi:adaptive_broadcast_2009,sadeghi:shams:traskov:2010}, prohibits the transmissions of non-instantly decodable packets to any receiver. Effectively, each coded packet can include at most one wanted data packet of every receiver. The second one, called general IDNC (G-IDNC), removes this restriction to generate more coding opportunities.



There is a large body of research on G-IDNC, focusing on its throughput and decoding delay performance, coding algorithms, and transmission schemes. Early models and heuristics related to G-IDNC were proposed for index coding \cite{sprintson:min:2007}. Then G-IDNC was introduced for wireless broadcast and was graphically modeled in \cite{sameh:valaee:globecom:2010}. Although its best performance remains unidentified, powerful heuristic algorithms have been developed to improve its throughput and/or decoding delay \cite{sameh:valaee:globecom:2010,sorour:valaee:2010,sameh:density:2013,le:realtime_IDNC}, or to strike a balance between the two \cite{parastoo:tradeoff:2014}. G-IDNC transmission schemes under full/intermittent receiver feedback have been developed \cite{sorour:lossy_feedback:2011,sorour:limited_feedback:2011}. G-IDNC has also been adopted in wireless broadcast applications with hard decoding deadlines \cite{li:idnc_video:2011} or with receiver cooperation \cite{neda:2013:cooperative, karim:2014:cooperative,keshtkarjahromi2014content}.


\begin{table}
\centering
\caption{A simple packet reception state for IDNC coding}
\label{tab:simple_sfm}
\begin{tabular}{|c|c|c|}
\hline
~&$\p_1$&$\p_2$\\\hline
$R_1$&has&wants\\\hline
$R_2$&wants&wants\\\hline
$R_3$&has&has\\\hline
\end{tabular}
\end{table}

Studies on theoretical performance characterization and implementations of S-IDNC  are more limited in both range and depth. S-IDNC was graphically modeled in \cite{Rozner_Heuristic_clique}, which then proved that the minimum clique partition solution \cite{Rozner_Heuristic_clique} of the associated graph can be an S-IDNC solution that minimizes the block completion time. However, this solution does not take into account the issues of decoding delay and the robustness of coded transmissions to erasures. S-IDNC has shown to be asymptotically throughput optimal when there are up to three receivers or when the number of data packets approaches infinity \cite{li:idnc_video:2011}, but the general relation between the throughput of S-IDNC and system parameters has not been characterized before. In addition and to the best of our knowledge, the minimum packet decoding delay of S-IDNC is still unknown. Moreover, there have not been S-IDNC transmission schemes that can work with intermittent receiver feedback. Another unaddressed problem is a systematic performance comparison between S-IDNC and G-IDNC.

In this paper, we study the above problems and provide the following contributions:
\begin{enumerate}
\item We characterize the throughput performance limits of S-IDNC. Specifically, we first derive an achievable upper bound on the minimum block completion time for any given packet reception state. We then derive a closed-form expression for the expected minimum block completion time in terms of the number of packets and receivers and the erasure probability;
\item We prove that it is NP-hard to minimize the packet decoding delay of S-IDNC. We derive an upper bound on the minimum packet decoding delay in terms of the minimum block completion time;
\item We show that in the presence of erasures, the minimum clique partition solution of the S-IDNC graph as identified by \cite{Rozner_Heuristic_clique} may not result in the minimum block completion time, because it does not allow the same data packet to be repeated in different coded packets, a desired property which we refer to as packet diversity. Motivated by this fact, we develop the optimal S-IDNC coding algorithm, as well as heuristics that aim to improve packet diversity. We design S-IDNC transmission schemes under full and intermittent receiver feedback;
\item We also provide new results on how S-IDNC and G-IDNC compare. We study the relation between S-IDNC and G-IDNC graphs, and then demonstrate some scenarios under which G-IDNC can/cannot outperform S-IDNC.
\end{enumerate}

\section{System Model and Notations}
\subsection{Transmission Setup}
We consider a block-based wireless broadcast scenario, in which the sender needs to deliver a block of $K$ data packets, denoted by $\P_K=\{\p_k\}_{k=1}^K$, to $N$ receivers, denoted by $\R_N=\{R_n\}_{n=1}^N$ through wireless channels that are subject to independent random packet erasures.

Initially, the $K$ data packets are transmitted uncoded once using $K$ time slots, constituting a \emph{systematic transmission phase} \cite{heide_systematic_RLNC}. Then, each receiver provides feedback\footnote{We assume that there exists an error-free feedback link from each receiver to the sender that can be used with appropriate frequency.} to the sender about the packets it has received or missed due to packet erasures. The complete packet reception state is represented by an $N \times K$ state feedback matrix (SFM) $\mA$, where $a_{n,k}=1$ if $R_n$ has missed (and thus still wants) $\p_k$, and $a_{n,k}=0$ if $R_n$ has already received $\p_k$. The set of data packets wanted by $R_n$ is called the \emph{Wants} set of $R_n$ and is denoted by $\W_n$. The set of receivers who want $\p_k$ is called the \emph{Target} set of $\p_k$ and is denoted by $\T_k$. The size of  $\T_k$ is denoted by $T_k$. Packets with larger $T_k$ are more desired by receivers.

\begin{Example}
Consider the SFM in \figref{fig:sfm_example} with $K=6$ data packets and $N=5$ receivers. The Wants set of $R_1$ is $\W_1=\{\p_1,\p_5,\p_6\}$. The Target set of $\p_3$ is $\T_3=\{R_3,R_5\}$ and thus $T_3=2$.
\end{Example}

\begin{figure*}
\centering
\subfigure[State feedback matrix $\mA$]{\includegraphics[width=0.23\linewidth]{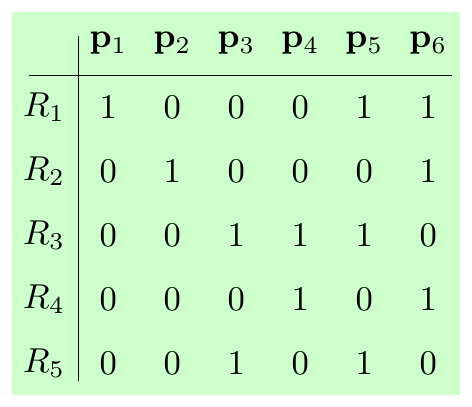}\label{fig:sfm_example}}\hspace{20pt}
\subfigure[S-Graph $\G_s$]{\includegraphics[width=0.21\linewidth]{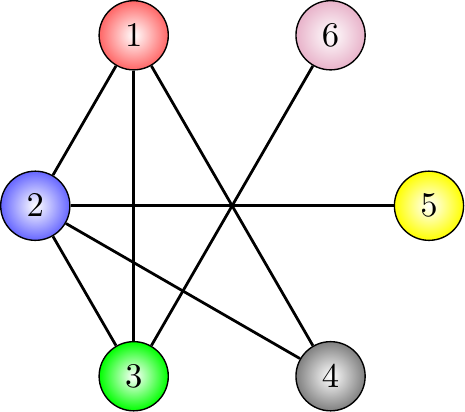}\label{fig:sg_example}}\hspace{20pt}
\subfigure[G-Graph $\G_g$]{\includegraphics[width=0.2\linewidth]{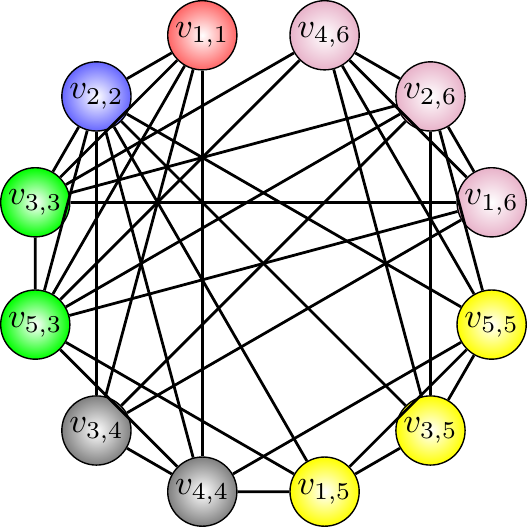}\label{fig:gg_example}}
\caption{An example of SFM and its S- and G-IDNC graphs.}
\label{fig:example}
\end{figure*}

\subsection{Coded Transmission Phase: Two Types of IDNC}
According to $\mA$, the sender generates IDNC coded packets under the binary field $\F_2$. Explicitly, IDNC coded packets are of the form $\X = \bigoplus_{\p_k\in\M} \p_k$, where $\M$ is a selected subset of $\P_K$, and is called an IDNC coding set. Obviously, $\X$ has three possible decoding effects at each receiver:
\begin{defn}
An IDNC coded packet $\X$ is instantly decodable for receiver $R_n$ if
$\M$ contains only one data packet from the Wants set $\W_n$ of $R_n$, i.e., if $|\M\cap\W_n|=1$.
\end{defn}
\begin{defn}
An IDNC coded packet $\X$ is non-instantly decodable for receiver $R_n$ if $\M$ contains two or more data packets from the Wants set $\W_n$ of $R_n$, , i.e., if $|\M\cap\W_n|>1$.
\end{defn}
\begin{defn}
An IDNC coded packet $\X$ is non-innovative for receiver $R_n$ if $\M$ contains no data packets from the Wants set $\W_n$ of $R_n$, , i.e., if $|\M\cap\W_n|=0$. Otherwise, it is innovative.
\end{defn}

Depending on which of the above three types of effects are allowed, there are two variations of IDNC. The first one is called strict IDNC (S-IDNC), which is the main subject of our study. It prohibits the transmission of any non-instantly decodable coded packets to any receivers. This restriction implies that any two data packets wanted by the same receiver cannot be coded together. We thus have the concept of conflicting and non-conflicting data packets:
\begin{Definition}
Two data packets $\p_i$ and $\p_j$ conflict if at least one receiver wants both of them, i.e., if $\exists n:\{\p_i,\p_j\}\subseteq\W_n$. Otherwise they do no conflict.
\end{Definition}

An S-IDNC coding set is thus a set of pairwise non-conflicting data packets. The conflicting states among data packets can be represented by an undirected graph $\G_s(\V,\E)$. Each vertex $\v_i\in\V$ represents a data packet $\p_i$. Two vertices $\v_i$ and $\v_j$ are connected by an edge $\e_{i,j}\in\E$ if $\p_i$ and $\p_j$ do not conflict. Thus, every complete subgraph of $\G_s$, a.k.a., a clique, represents an S-IDNC coding set. In the rest of the paper, we will use the terms ``coded packet'', ``coding set'', and ``clique'' interchangeably, and denote the last two by $\M$.

The main limitation of S-IDNC is that a coded packet which is instantly decodable for a large subset of receivers may be prohibited merely because it is non-instantly decodable for a small subset of receivers. In the second type of IDNC, called general IDNC (G-IDNC), the restriction on non-instantly decodable packets is removed to generate more coding opportunities.

G-IDNC can also be graphically modeled \cite{sorour:valaee:2010}. The difference is that, in the G-IDNC graph $\G_g(\V,\E)$, a data packet $\p_k$ wanted by different receivers are individually represented by different vertices $\v_{n,k}$, for all $a_{n,k}=1$. Consequently, the number of vertices in $\G_g$ is equal to the number of ``1''s in $\mA$. Two vertices $\v_{m,i}$ and $\v_{n,j}$ are connected by an edge if: 1) $i=j$, or 2) if $\p_i\notin\W_n$ and $\p_j\notin\W_m$. In the first case, $\p_i=\p_j$, and thus by sending $\p_i$ both $R_m$ and $R_n$ can decode. In the second case, by sending $\p_i\oplus\p_j$, $R_m$ and $R_n$ can decode $\p_i$ and $\p_j$, respectively, because they already have $\p_j$ and $\p_i$, respectively. Similar to S-IDNC, every clique of $\G_g$ represents a G-IDNC coding set.

We note that an S-IDNC coded packet is always a G-IDNC coded packet, but the reverse is not necessarily true. Below is an example of S- and G-IDNC coded packets.

\begin{Example}
Consider the SFM and its S- and G-IDNC graphs in \figref{fig:example}. The G-IDNC graph indicates that $(\v_{1,1},\v_{5,3},\v_{4,4})$ is a clique. The corresponding G-IDNC coding set is $(\p_1,\p_3,\p_4)$, and thus $\X_g=\p_1\oplus\p_3\oplus\p_4$ is a G-IDNC coded packet. $\X_g$ is instantly decodable for $R_{1}, R_4, R_5$ because they only want one data packet from $\X_g$. $\X_g$ is non-instantly decodable for $R_3$ because $R_3$ wants both $\p_3$ and $\p_4$. $\X_g$ is non-innovative for $R_2$.

Due to the existence of $R_3$, $\X_g$ is not an S-IDNC coded packet. Whereas the S-IDNC graph indicates that $(\v_{1},\v_{2},\v_{3})$ is a clique. The corresponding coding set is $(\p_1,\p_2, \p_3)$, and thus $\X_s=\p_1\oplus\p_2\oplus\p_3$ is an S-IDNC coded packet, which can be verified to also correspond to clique $(\v_{1,1},\v_{2,2},\v_{3,3})$ or clique $(\v_{1,1},\v_{2,2},\v_{5,3})$  in the G-IDNC graph.
\end{Example}

We then introduce the notion of \emph{IDNC solution}. A set of IDNC coding sets is called an IDNC solution if, upon the reception of the coded packets of all these coding sets, every receiver can decode all its wanted data packets. An S-IDNC solution is denoted by $\S_s$. The set of all S-IDNC solutions of a given SFM is denoted by $\mathbb S_s$. Similarly, we can also define $\S_g$ and $\mathbb S_g$ for G-IDNC.

For the SFM in \figref{fig:example}, by partitioning the S-IDNC graph into three disjoint cliques, we can obtain, among others, an S-IDNC solution of $\S_s=\{(\p_1,\p_4),(\p_2,\p_5),(\p_3,\p_6)\}$. Similarly, a disjoint clique partition of the G-IDNC graph is $ \{(\v_{1,1},\v_{2,2},\v_{5,3},\v_{4,4})$, $(\v_{3,3},\v_{1,6},\v_{2,6},\v_{4,6})$, $(\v_{1,5},\v_{3,5},\v_{5,5})$, $(\v_{3,4},\v_{4,4})\}$, indicating a G-IDNC solution of $\S_g =\{(\p_1,\p_2,\p_3,\p_4),(\p_3,\p_6),(\p_5),(\p_4)\}$.

To assess the performance of IDNC solutions, we now introduce our measures of throughput and decoding delay.

\subsection{Throughput and Decoding Delay Measures}
An S-IDNC solution $\S_s$ requires a minimum of $|\S_s|$ coded transmissions. We call $U_{\S_s}\triangleq |\S_s|$ the \emph{minimum block completion time} of $\S_s$. It measures the best throughput of $\S_s$ with a value of $\frac{K}{K+U_{\S_s}}$ packet per transmission. We further denote by $U_s$ the absolute minimum block completion time over all the S-IDNC solutions of $\mA$, i.e., $U_s\triangleq \min\{U_{\S_s}:\S_s\in\mathbb S_s\}$. The definition of $U_g$ for G-IDNC follows.

We measure decoding delay by \emph{average packet decoding delay} (APDD) $D$, which is the average time it takes for a receiver to decode a data packet:
\begin{equation}\label{eq:d_def}
D=\frac{1}{T}\sum_{\forall a_{n,k}=1}u_{n,k}
\end{equation}
where $u_{n,k}$ is the time index when $R_n$ decodes $\p_k$, and $T=\sum_{k=1}^KT_k$, which is also the number of ``1''s in $\mA$.
Given an IDNC solution $\S$, by letting $u_{n,k}$ be the \emph{first} time index when $\S$ allows $R_n$ to decode $\p_k$, \eqref{eq:d_def} produces the minimum APDD of $\S$. Then similar to $U_s$ (resp. $U_g$), we denote by $D_s$ (resp. $D_g$) the absolute minimum APDD over all S- (resp. G-) IDNC solutions of $\mA$. We also note that in the specific case of an S-IDNC solution $\S_s$, $u_{n,k}$ is indeed the index of the first coding set in $\S_s$ that contains $\p_k$, as every receiver who wants $\p_k$ can decode it from this coding set.

\begin{Example}
Consider the SFM in \figref{fig:sfm_example}. Suppose that an S-IDNC solution with four coded packets $\X_1=\p_1\oplus\p_2$, $\X_2=\p_3\oplus\p_6$, $\X_3=\p_4$, and $\X_4=\p_5$ are transmitted in this order. The receivers' decoding time $\{u_{n,k}\}$ are summarized in Table \ref{tab:delay_sfm}. The minimum APDD of this solution is $D_{\S_s}=(1\times2+2\times5+3\times2+4\times3)/12=2.5$.

\begin{table}[t]
\centering
\caption{The decoding delay of original data packets at the receivers}
\vspace{-1em}
\begin{tabular}{c|cccccc}
~    &$\p_1$&$\p_2$&$\p_3$&$\p_4$&$\p_5$&$\p_6$\\\hline
$R_1$& 1    & 0    & 0    & 0    & 4    & 2\\\hline
$R_2$& 0    & 1    & 0    & 0    & 0    & 2\\\hline
$R_3$& 0    & 0    & 2    & 3    & 4    & 0\\\hline
$R_4$& 0    & 0    & 0    & 3    & 0    & 2\\\hline
$R_5$& 0    & 0    & 2    & 0    & 4    & 0\\\hline
\end{tabular}
\label{tab:delay_sfm}
\end{table}

\end{Example}

In the presence of packet erasures, the sender needs to adopt a \emph{coded transmission phase}. In each coded transmission, it selects and broadcasts a coding set through erasure-prone wireless channels. We denote by $U_T$ the block completion time of this phase, and by $D_T$ the APDD of this phase, calculated as in \eqref{eq:d_def}. $U_T$ and $D_T$ measure the throughput and decoding delay performance of this phase, respectively. They vary according to the IDNC solutions, transmission schemes, and erasure patterns. But it always holds that $U_T\geqslant U_s$ and $D_T\geqslant D_s$ if S-IDNC is applied. Therefore, $U_{s}$ and $D_{s}$ reflect the performance limits of S-IDNC. Hence, we will first study these limits in the next section, and then design S-IDNC transmission schemes and coding algorithms in Sections \ref{sec:schemes} and \ref{sec:algorithms}, respectively.

\section{Performance Limits and Properties of IDNC}\label{sec:limits}
In this section, we study performance limits and properties of S-IDNC and compare it with G-IDNC.

\subsection{Absolute minimum block completion time $U_s$}
We first study the throughput limit of S-IDNC, measured by the absolute minimum block completion time $U_s$. It has been proved that $U_s$ is equal to the size of the minimum clique partition solution\footnote{The minimum clique partition solution of a graph $\G$ is the minimum set of disjoint cliques of $\G$ that together cover all the vertices.} of $\G_s$ \cite{Rozner_Heuristic_clique}, denoted by $\S_c$. This equivalence holds because of the following property:
\begin{property}\label{prpt:vertex_remove}
Removing any vertex from the S-IDNC graph does not change the connectivity of the remaining vertices.
\end{property}
This property holds because vertices in $\G_s$ represent different data packets. Thus, to remove all vertices from $\G_s$ (i.e., to complete the broadcast), at least $|\S_c|$ cliques must be removed, which yields $U_s=|\S_c|$.

According to graph theory, $|\S_c|$ is equal to the chromatic number\footnote{The chromatic number of a graph $\G$ is the minimum number of colors to color the vertices so that any two connected vertices have different colors.} $\chi(\cG_s$) of the complementary graph $\cG_s$, which has the same vertex set as $\G_s$, but has opposite vertex connectivity. We thus have $U_s=\chi(\cG_s)$. This equality enables us to answer two important questions about $U_s$: 1) how to find the $U_s$ of a given S-IDNC graph? 2) what are the statistical characteristics of $U_s$ under random packet erasures?

\subsubsection{The $U_s$ of an SFM}

The chromatic number of a graph (and thus $U_s$) has been proven to be NP-hard to find and AXP-hard to approximate \cite{Graph_theory}. But there are heuristic algorithms and bounds for it. We will develop algorithms dedicated to S-IDNC in Section \ref{sec:algorithms}, and focus on the bounds in this subsection.

Tight bounds on $\chi(\cG_s)$ exist, but are also NP-hard to find. One such example is a tight lower bound $w(\cG_s)$ \cite{Graph_coloring}, the size of the largest clique of $\cG_s$. There are also loose bounds. For example:
\begin{property}
All S-IDNC graphs with $K$ vertices and $M_0$ edges have:
\begin{equation}
U_s\geqslant\left\lceil K^2/(K+2M_0)\right\rceil
\end{equation}
where $\lceil x\rceil$ outputs the smallest integer greater than $x$.
\end{property}
This bound is due to Geller \cite{Geller_lower_bound} and its proof is omitted here. This bound is useful because it identifies the smallest \emph{achievable} $U_s$ of all the S-IDNC graphs with $K$ vertices and $M_0$ edges.

An existing upper bound on $\chi(\cG_s)$ is $\Delta(\cG_s)+1$ \cite{Graph_theory}, where $\Delta(\cG_s)$ is the largest number of edges incident to any single vertex in $\cG_s$. However, with given $K$ and $M_0$, this bound is not always achievable. To see this, assume $\cG_s$ has $K-1$ edges. Connecting these edges to the same vertex yields an upper bound of $\Delta(\cG_s)+1=K$. But no matter how we allocate these edges, there are always unconnected vertices in $\cG_s$, which indicates that $\chi(\cG_s)<K$. Hence, the upper bound is not achievable here.

We are thus motivated to derive an achievable upper bound of $U_s$ as a function of $K$ and $M_0$. We first note that a set of pairwise unconnected vertices (a.k.a. an independent set) of $\G_s$, denoted by $\V_I$, must be transmitted separately because their corresponding packets all conflict with each other. The size of $\V_I$ is $K$ when there is no edge in $\G_s$, indicating that $U_s=|\V_I|=K$. Then, whenever a new edge is added to the graph, we can maximize $U_s$ by maximizing $|\V_I|$, which means that the edge should not connect two vertices in $\V_I$ whenever possible. Explicitly, our upper bound on $U_s$ is derived iteratively:
\begin{itemize}
\item When there is no edge in the graph, we have $U_s=K$;
\item When there are $M_0=[1,K-1]$ edges, we can use them to connect $\v_1$ with $\v_2\cdots\v_K$. Since $\v_2\cdots\v_K$ are independent, we have $U_s=K-1$;
\item When there are up to $K-2$ additional edges, i.e, when $M_0=[K,2K-3]$, we can use these additional edges to connect $\v_2$ with $\v_3\cdots\v_K$. Since $\v_3\cdots\v_K$ are independent, we have $U_s=K-2$;
\item The iterations will terminate when the graph is complete, i.e., when $M_0=K(K-1)/2$ and  $U_s=1$.
\end{itemize}

It can be easily proved that any reallocation of edges will reduce the $U_s$ derived above to a smaller value. Our upper bound on $U_s$ has the following stair-case profile:
\begin{property}
All S-IDNC graphs with $K$ vertices and $M_0$ edges have:
\begin{equation}\label{eq:U_upper}
U_s\leqslant\begin{cases}
K,&M_0=0\\
K-1,&M_0=[1,K-1]\\
K-2,&M_0=[K,2K-3]\\
\cdots,&\cdots\\
1,&M_0=K(K-1)/2
\end{cases}
\end{equation}
\end{property}

\subsubsection{$U_s$ as a function of system parameters}

In addition to finding $U_s$ for a given S-IDNC graph, we are also interested in the statistical characteristics of $U_s$, for which we assume that $\mA$ (and thus $\G_s$) is obtained as a consequence of random packet erasures in the systematic transmission phase.

For wireless broadcast, a common assumption on random packet erasures is that they are i.i.d. random variables with a probability of $P_e$. Under this assumption, a similar question has already been introduced and answered for the RLNC technique. It has been shown in \cite{eryilmaz:ozdaglar:medard:ahmed:2008,xiao2010extreme,ghaderi2008reliability} that the block completion time of RLNC scales as $\mathcal O(\ln(N))$ when $K$ is a constant. Consequently, the throughput of RLNC vanishes with increasing number of receivers $N$. To prevent zero throughput, it has been proved in \cite{swapna2013throughput} that $K$ should scale faster than $\ln(N)$.

Since the throughput of RLNC is optimal, it cannot be exceeded by the throughput of S-IDNC. Hence, we can infer that the throughput of S-IDNC should also follow a vanishing behavior with increasing $N$. However, its rate and specific dependence on system parameters have not been fully characterized in the literature. In this subsection, we answer this question through the following theorem:
\begin{Theorem}\label{theo:Us}
The mean of the absolute minimum block completion time $U_s$ is a function of the block size $K$, the number of receivers $N$, and packet erasure probability $P_e$ as follows:
\begin{align}
E[U_s]&=\left[-K\left(\frac{1}{2}+o(1)\right)\log_K(1-P_e^2)\right]\cdot N\label{eq:u_parameters}\\
&=c(K,P_e)\cdot N
\end{align}
where $o(1)$ is a small term that approaches zero with increasing $K$.
\end{Theorem}
\begin{IEEEproof}
Our approach is to model the complementary S-IDNC graph $\cG_s$ as a random graph with i.i.d. edge generating probability. Recall that two vertices in $\cG_s$ are connected if the two data packets conflict, i.e., if at least one receiver has missed both packets. Therefore, the edge generating probability, denoted by $P_c$, is calculated as:
\begin{equation}\label{eq:p_conflict}
P_c=1-(1-P_e^2)^N,
\end{equation}

Then, the key is to prove that different edges are generated independently. We first consider the independency between two adjacent edges, say $\e_{1,2}$ and $\e_{1,3}$, which share $\v_1$. The information carried by $\e_1$ (resp. $\e_2$) is that there is at least one receiver who wants both $\p_1$ and $\p_2$ (resp. $\p_1$ and $\p_3$). Hence, the mutual information between $\e_1$ and $\e_2$ is that $\p_1$ is wanted by at least one receiver, which happens with a probability of $1-(1-P_e)^N$. Thus:
\begin{equation}
I(\e_{1,2};\e_{1,3})\leqslant H((1-P_e)^N,1-(1-P_e)^N)\triangleq H(\p_1)
\end{equation}
The inequality holds because other edges incident to $\p_1$ also contribute to $H(\p_1)$. It is clear that $I(\e_{1,2};\e_{1,3})$ quickly converges to zero with increasing $N$, indicating that $\e_{1,2}$ and $\e_{1,3}$ are asymptotically independent of each other. We note that two disjoint edges in $\cG_s$ share no mutual information, and thus are mutually independent. Therefore, we can assume that edges in $\cG_s$ are independently generated.

Consequently, $\cG_s$ can be modeled as an Erd\~os-R\'enyi random graph \cite{Random_graph}, which has $K$ vertices and i.i.d. edge generating probability of $P_c$. \figref{fig:M0} compares the mean number of edges (with a value of $K(K-1)/2\cdot P_c$) of our proposed random graph model and the simulated average number of edges in $\cG_s$. Our model shows virtually no deviation under all considered values of $N$ and $K$.
\begin{figure}
\centering
\includegraphics[width=0.6\linewidth]{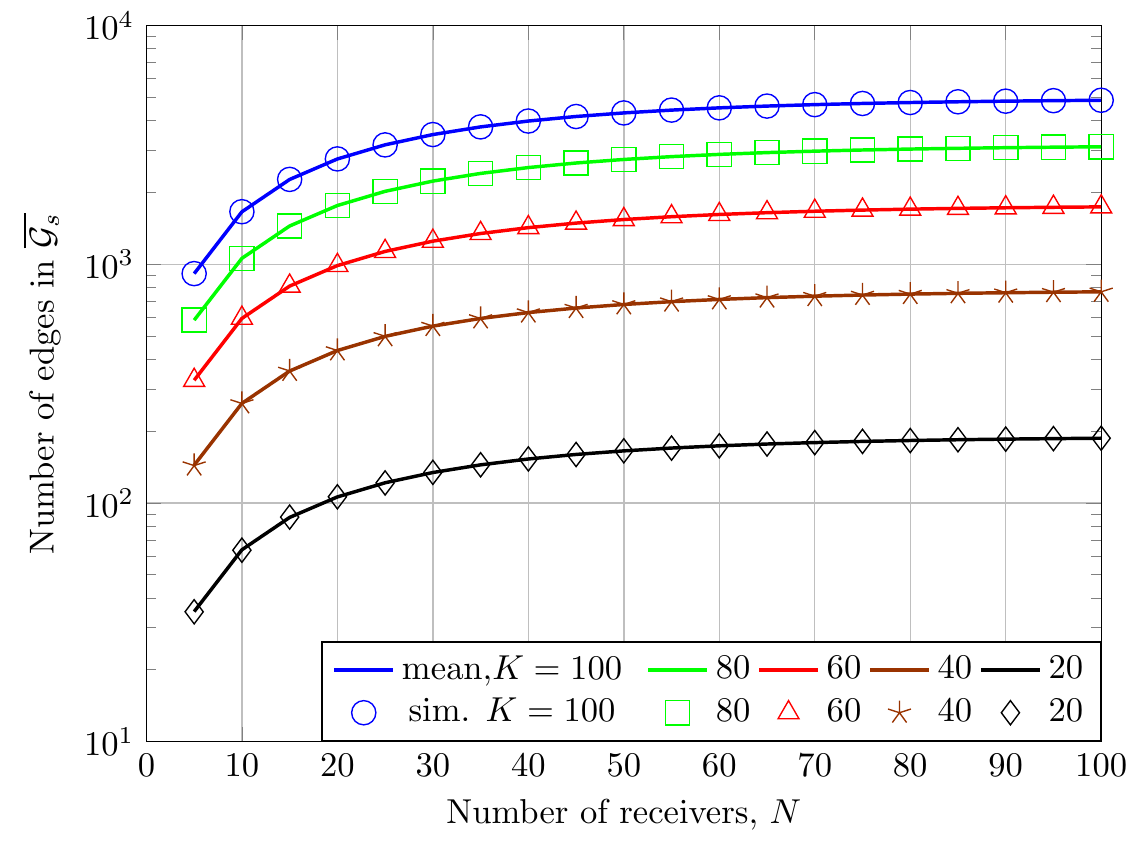}
\caption{The mean and simulated number of edges in $\cG_s$ when $P_e=0.2$ and $K\in[20,100]$.}
\label{fig:M0}
\end{figure}
From graph theory, given $K$ and $P_c$, almost every random graph $\cG_s$ has a chromatic number of:
\begin{equation}\label{eq:random_chromatic}
\chi(\cG_s)=\frac{K}{\log K}\left(\frac{1}{2}+o(1)\right)\log\frac{1}{1-P_c}
\end{equation}
Since $U_s=\chi(\cG_s)$, the above value is the mean of $U_s$. By substituting \eqref{eq:p_conflict} into \eqref{eq:random_chromatic} we obtain \eqref{eq:u_parameters}.
\end{IEEEproof}
Theorem \ref{theo:Us} has the following important corollary:
\begin{Corollary}\label{cor:U_linear_N}
The mean of the absolute minimum block completion time, $E[U_s]$, of S-IDNC increases linearly with the number of receivers in wireless broadcast with i.i.d. random packet erasures.
\end{Corollary}
Then, by noting that the mean block completion time of the coded transmission phase is lower bounded by $E[U_s]$, we conclude that the throughput of S-IDNC is significantly affected by the number of receivers. S-IDNC may not be a good choice when the number of receivers is large. We note that the above theorem and corollary are not directly applicable to G-IDNC, because the edge generating probability in G-IDNC is quite different. Interested readers are referred to \cite{sorour:valaee:2011} for more information.

\subsection{Absolute minimum average packet decoding delay $D_s$}
Unlike $U_s$, to the best of our knowledge there is no existing hardness result on finding $D_s$. In this subsection, we address it through the following theorem and then propose an upper bound on $D_s$.

\begin{Theorem}\label{theo:D_hard}
It is NP-hard to find the absolute minimum APDD $D_s$ of S-IDNC.
\end{Theorem}

\begin{proof}
In order to prove it, we introduce the concept of \emph{perfect S-IDNC solution}:
\begin{Definition}
An S-IDNC solution is perfect and is denoted by $\S_p$ if every receiver $R_n$ can decode one of its $|\W_n|$ wanted data packets from each of the first $|\W_n|$ coding sets in $\S_p$.
\end{Definition}

This definition implies that $\S_p$, if it exists, offers the perfect packet decoding. Hence, its minimum APDD $D_{\S_p}$ is not only the absolute minimum APDD $D_s$ of S-IDNC, but also a lower bound $\underline{D}$ of the minimum APDD of all other linear network coding techniques. Moreover, it can be easily shown that this lower bound can only be achieved by $\S_p$.

It has been proved in \cite{yu:alex:parastoo:delay:netcod} that it is NP-hard to determine the achievability of $\underline{D}$. Hence, it is NP-hard to determine the existence of $\S_p$. Then by contradiction, if it is not NP-hard to find $D_s$, we can easily determine the existence of $\S_p$ by comparing $D_s$ with $D_{\S_p}$. Therefore, it is NP-hard to find $D_s$.
\end{proof}
Besides the NP-hardness, $D_s$ has the following property:
\begin{property}
The absolute minimum APDD $D_s$ is upper bounded by $U_s$ as
\begin{equation}\label{eq:d_us}
D_s\leqslant\frac{U_s+1}{2}
\end{equation}
\vspace{-2em}
\label{propt:d_bound}
\end{property}
\begin{proof}
Given an S-IDNC solution $\S_s=\{\M_u\}_{u=1}^{U}$, let $T(u)=\sum_{\p_k\in\M_u}T_k$ be the number of receivers who can decode a data packet from $\M_u$. The minimum APDD of $\S_s$ is thus:
\begin{equation}\label{eq:d_order}
D_{\S_s}=\frac{1}{T}\sum_{u=1}^{U}T(u)\cdot u
\end{equation}
which is maximized when $\{T(u)\}_{u=1}^U=\frac{T}{U}$. In this case, $D_{\S_s}=\frac{U+1}{2}$. Applying this result to an S-IDNC solution with absolute minimum block completion time $U = U_s$, we obtain the result.
\end{proof}

Our proof indicates that, although it is NP-hard to achieve $D_s$, we can still effectively reduce APDD by reducing the S-IDNC solution size. Before we further explore this result to implement S-IDNC, we would like to compare the performance limits of S-IDNC that we have just derived with G-IDNC.

\subsection{S-IDNC vs. G-IDNC}\label{sec:sg_comp_erasure_free}
In this subsection, we address the question of \emph{how does S-IDNC compare with G-IDNC?}

We first note that the NP-hardness of finding $D_s$ also holds for $D_g$. This is because the perfect S-IDNC solution $\S_p$ is also the best possible G-IDNC solution. For the throughput, we first present a relation between S- and G-IDNC graphs (proved in the appendix):
\begin{Theorem}\label{theo:equal_solutions}
The minimum clique partition solutions of S-IDNC and G-IDNC graphs have the same size. In other words, $\chi(\cG_s)=\chi(\cG_g)$.
\end{Theorem}

However, the above theorem does not imply $U_s=U_g$. This is because G-IDNC does not have Property \ref{prpt:vertex_remove}. Explicitly, by removing a vertex from $\G_g$, more edges and larger cliques may be generated, and thus the absolute minimum block completion time $U_g$ can be smaller than $\chi(\cG_g)$ of the original G-IDNC graph $\G_g$ \cite{sameh:density:2013}.  We thus have $U_g\leqslant U_s$. We note, however, that a systematic way of finding $U_g$ other than brute-force search remains widely open.

Therefore, when there are no packet erasures, the throughput of G-IDNC is at least as good as S-IDNC. But is this still true in more realistic erasure-prone scenarios? In the next section, we will design S-IDNC transmission schemes under packet erasures and compare them with G-IDNC. We will apply the above theorem to show that G-IDNC cannot outperform S-IDNC under certain circumstances.

\section{S-IDNC Transmission Schemes}\label{sec:schemes}
In this section, we design S-IDNC transmission schemes to compensate for packet erasures in the coded transmission phase, which are i.i.d. with a probability of $P_e$. To this end, the sender has to regularly collect feedback from the receivers about their packet reception state to make online coding decisions. We consider two common types of feedback frequency, namely:
\begin{enumerate}
\item fully-online feedback: feedback is collected after every coded transmission. However, this could be costly in wireless communications. We thus also consider a reduced feedback frequency next;
\item semi-online feedback: feedback is only collected after several (to be quantified later) coded transmissions;
\end{enumerate}

To be able to design S-IDNC transmission schemes, two questions need to be answered first:
\begin{enumerate}
\item What is the optimization objective for throughput and decoding delay improvement?
\item What does the sender need to send to achieve it?
\end{enumerate}

Before addressing these questions, we first highlight some challenges:

\begin{Remark}\label{remark:tractable_matrics}
Under random packet erasures, a reasonable measure of throughput is the mean block completion time $E[U_T]$ of the coded transmission phase. However, it is intractable to minimize $E[U_T]$. To see this, let us consider the stochastic shortest path (SSP) method \cite{sorour:valaee:2010}.
In SSP method, the state space comprises the current SFM and its successors, and thus has a prohibitively large size with a value of $2^T$, where $T$ is the number of ``1''s in $\mA$. The action space for each state comprises all cliques/coding sets, which is NP-hard to find \cite{Bron_algorithm}. Then, $E[U_T]$ is recursively minimized by examining all the states and the associated actions. Such examination is necessary, because the packet erasures can take any pattern and are not predictable. But it makes $E[U_T]$ intractable to minimize. To overcome this difficulty, we will turn to optimization objectives that are heuristic, but still based on SSP optimization principles.
\end{Remark}

\begin{Remark}\label{remark:U_prior}
It is intractable to minimize the APDD $D_T$ of the coded transmission phase due to the NP-hardness of finding $D_s$, because otherwise by setting $P_e=0$, the minimum $D_T$ is equal to $D_s$. To overcome this difficulty, we will give higher priority to the minimization of block completion time. In other words, we first minimize the block completion time. Then among the resultant coding decisions, we choose the one that minimizes the decoding delay. Our prioritization reflects the motivation of using network coding, that is, to achieve better throughput performance. It also provides bounded decoding delay performance as we have shown in \eqref{eq:d_us}. This will also be confirmed by our simulations, which show that $D_T$ generally decreases with decreasing $U_T$.
\end{Remark}

\subsection{Fully-online Transmission Scheme}
With fully-online feedback frequency, the sender only transmits one coded packet before collecting feedback. Under the SSP method, the current state is the current SFM $\mA$, the absorbing state is the all-zero SFM and is denoted by $\mA_0$. The action space comprises all the S-IDNC coding sets of $\mA$. The cost of each action is one, for it consumes one transmission. The block completion time $U_T$ is thus equal to the number of transitions (a.k.a. path length or distance) between $\mA$ and $\mA_0$.

According to Remark~1, it is intractable to choose an action/coded packet that minimizes the expected path length (and thus $E[U_T]$). As a heuristic alternative, we propose to choose an action/coded packet that belongs to the shortest path from $\mA$ to $\mA_0$, which has a length of $U_s$. This choice guarantees that, upon the reception of the coded packet at all interested receivers, the shortest distance between the updated state $\mA'$ and $\mA_0$ is minimized to $U_s-1$. To this end, the coded packet must belong to a minimum clique partition solution $\S_c$. Otherwise, the shortest distance between $\mA'$ and $\mA_0$ is still $U_s$.

We then reduce APDD by forcing the coded packet to be maximal (and thus serving the maximal number of receivers). However, cliques in a minimum clique partition solution are not necessarily maximal. Hence, we further require the coded packet to belong to a set of $U_s$ maximal cliques that together cover all the data packets. This set is also an S-IDNC solution and is denoted by $\S_m$.

In conclusion, we propose the following coded packet $\M_f$ for fully-online transmission  scheme:

\emph{Given an SFM instance, the preferred coded packet $\M_f$ is the most wanted coded packet in $\S_m$, where $\S_m$ is an S-IDNC solution that contains $U_s$ maximal cliques.}

\subsection{Semi-online Transmission Scheme}
The fully-online transmission scheme is costly, not only in collecting feedback, but also in computational load, as it has to find $\S_m$ in every time slot. These problems can be mitigated by partitioning the coded transmission phase into rounds. In each round, the sender transmits a complete S-IDNC solution and only collects feedback  at the end of each round. We call this scheme the semi-online scheme.

Under the SSP method, the action space is the set of all S-IDNC solutions $\mathbb S_s$, and the cost of each action is the solution size $|\S_s|$, which is equal to the length of a semi-online transmission round. The total cost is thus equal to the block completion time.

According to Remark 1, it is intractable to minimize the expected total cost (and thus $E[U_T]$). As a heuristic alternative, we propose to minimize the expected cost of the shortest path between $\mA$ and $\mA_0$. The shortest path has a length of one, representing the event that every coded packet of the chosen solution $\S_s$ is received by all the interested receivers after only one semi-online round. Denote the probability of this event by $P_s$. Then the expected cost is $|\S_s|/P_s$, where $P_s$ is calculated as:
\begin{equation}\label{eq:p_s}
P_s=\prod_{k=1}^K(1-P_e^{d_k})^{T_k}
\end{equation}
Here $d_k$ is called the packet diversity and is defined below.
\begin{Definition}
The diversity $d_k$ of data packet $\p_k$ is the number of coding sets in $\S_s$ that comprise $\p_k$.
\end{Definition}

We note that the minimum clique partition solution $\S_c$ is not a preferred semi-online S-IDNC solution. Although $\S_c$ offers the smallest solution size ($|\S_c|=U_s$), it does not maximize $P_s$ because every data packet has a diversity of only one due to disjoint cliques in $\S_c$. Instead, the $\S_m$ we have proposed for the fully-online case can offer a higher $P_s$ than $\S_c$ due to possibly overlapping maximal cliques, while also offering the smallest solution size.

We still wish to answer the following question before choosing $\S_m$ as our preferred semi-online S-IDNC solution: \emph{Is there a solution that, though large in its size, provides higher packet diversities, so that $P_s$ is maximized?}

An explicit answer to this question is difficult to obtain, because it requires the examination of all the solutions of size greater than $U_s$. Such search is costly and does not provide any insight into this question. Moreover, a solution with a larger block completion time is unlikely to provide higher packet diversities due to the following property of S-IDNC solutions:
\begin{property}
Every coding set in an S-IDNC solution comprises at least one data packet with a diversity of one.
\end{property}
This property holds because if every data packet in a coding set has a diversity of greater than one, then this coding set can be removed from the solution without affecting the completeness of the solution. Due to the above property, an S-IDNC solution $\S_s$ has at least $|\S_s|$ data packets with a diversity of only one. According to \eqref{eq:p_s}, these unit-diversity data packets reduce $P_s$ the most.

Therefore, we choose $\S_m$ for throughput improvement. Then, by taking into account our secondary optimization objective, i.e., the APDD, we define our preferred semi-online S-IDNC solution as follows:

\emph{Given an SFM instance, the proposed semi-online S-IDNC solution is $\S_m$, which comprises a set of $U_s$ maximal cliques. The cliques are sorted for transmission in the descending order of their numbers of targeted receivers to minimize the APDD.}

\begin{figure*}
\centering
\includegraphics[width=0.6\linewidth]{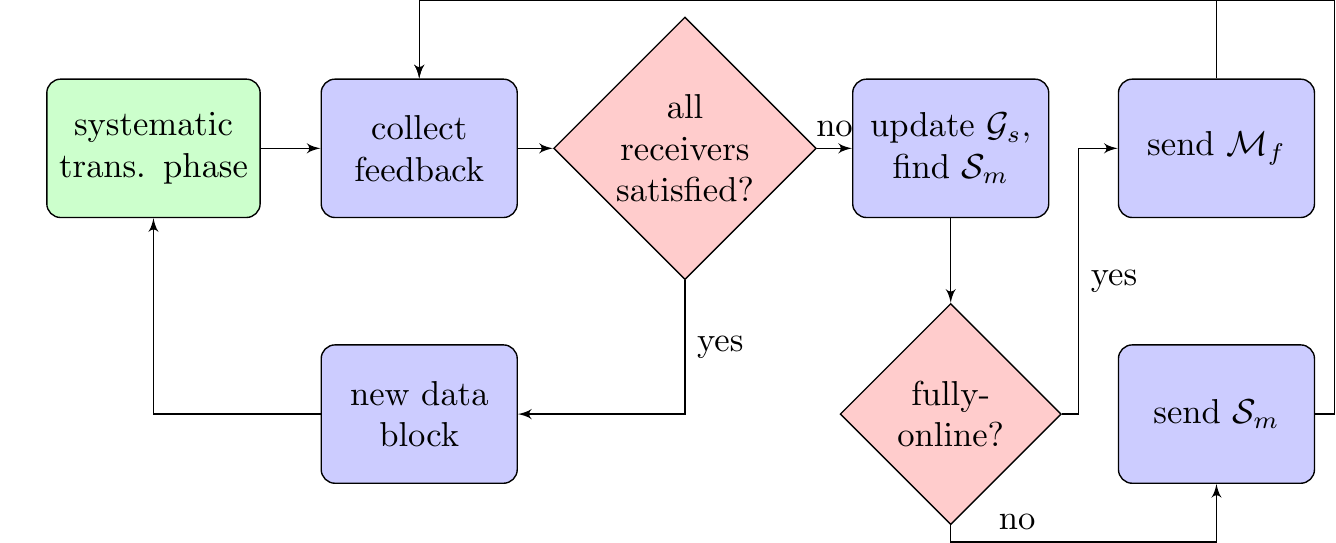}
\caption{The fully- and semi-online transmission schemes.}
\label{fig:flow}
\end{figure*}
A flow-chart of the proposed two transmission schemes are presented in \figref{fig:flow}. Both the fully- and semi-online IDNC schemes require finding $\S_m$. Since packet diversity is not a concern in graph theory, there is no algorithms to find $\S_m$ in the graph theory literature. Hence, we will design algorithms dedicated for S-IDNC in the next section. Before doing so, however, we briefly compare S-IDNC and G-IDNC under the above two transmission schemes.

\subsection{S-IDNC vs. G-IDNC}

With fully-online feedback, the sender can update the G-IDNC graph and add new edges representing coding opportunities after every transmission. The throughput of G-IDNC is thus better than S-IDNC. But the price is high computational load, because G-IDNC graph is much larger than S-IDNC graph ($\mathcal O(NK)$ v.s. $K$). However, during a semi-online transmission round, the sender cannot update SFM due to the absence of receiver feedback. Consequently, it does not update the G-IDNC graph $\G_g$ \cite{sorour:limited_feedback:2011}, and only sends the minimum clique partition solution of $\G_g$, which, according to Theorem \ref{theo:equal_solutions}, has the same size as the minimum clique partition solution of S-IDNC. We thus have the following corollary:
\begin{Corollary}
G-IDNC cannot reduce the length of a semi-online transmission round compared to S-IDNC.
\end{Corollary}

\section{S-IDNC Coding Algorithms}\label{sec:algorithms}

The two transmission schemes we proposed in the last section require finding $\S_m$, an S-IDNC solution that contains $U_s$ maximal coding sets. In this section, we develop its optimal and heuristic algorithms.

\subsection{Optimal S-IDNC coding Algorithm}\label{sec:optimal_algorithms}

Our optimal S-IDNC coding algorithm finds $\S_m$  in two steps:
\begin{enumerate}[{Step}-1]
\item \emph{Find all the maximal coding sets (maximal cliques):}
This problem is NP-hard in graph theory. We apply an exponential algorithm, called Bron-Kerbosch (B-K) algorithm \cite{Bron_algorithm}. The group of all maximal cliques is denoted by $\mathcal{A}$.
\item \emph{Find $\S_m$ from $\mathcal{A}$:}
We propose a branching algorithm in Algorithm \ref{algorithm:branch}. The intuition behind this algorithm is that, if a data packet $\p_k$ belongs to $d_k$ maximal coding sets in $\mathcal{A}$, then one of these $d_k$ maximal coding sets must be included in $\S_m$ for the completeness of $\S_m$. In the extreme case where $d_k=1$, the sole maximal coding set that contains $\p_k$ must be included in $\S_m$. Below is an example of Algorithm \ref{algorithm:branch}.
\end{enumerate}

\begin{figure*}
\centering
\includegraphics[width=0.6\linewidth]{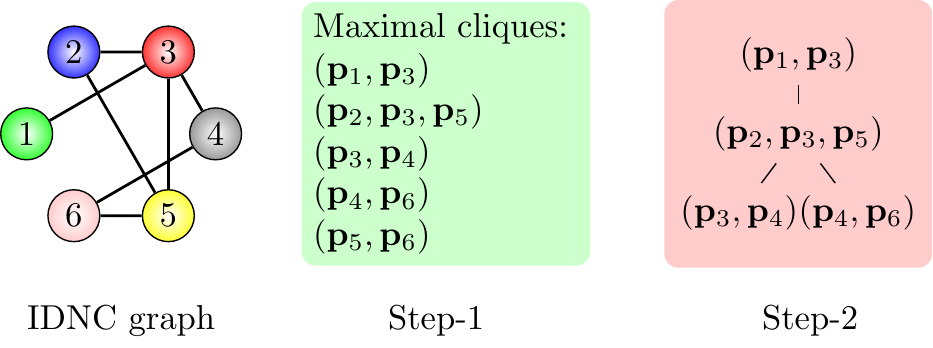}
\caption{An example of the 2-step optimal S-IDNC coding algorithm.}
\label{fig:min_collection}
\end{figure*}
\begin{Example}

Consider the graph model in \figref{fig:min_collection}. In Step-1, we find all the maximal cliques: $\mathcal{A}=\{(\p_1,\p_3)$, $(\p_2,\p_3,\p_5)$, $(\p_3,\p_4)$, $(\p_4,\p_6)$, $(\p_5,\p_6)\}$. Then in Step-2:
\begin{enumerate}
\item Initially, $\S=\emptyset$, $\overline{\S}=\mathcal{A}\setminus\S=\mathcal{A}$, and the set of data packets not included in $\S$ is $\overline{\P}=\{\p_1,\p_2,\p_3,\p_4,\p_5,\p_6\}$. Since $\p_1$ is only included in $(\p_1,\p_3)$ and $\p_2$ is only included in $(\p_2,\p_3,\p_5)$, these two coding sets must be added to $\S$. Hence, $\S=\{(\p_1,\p_3),(\p_2,\p_3,\p_5)\}$ after the first two iterations;
\item The set of data packets not included in $\S$ is $\overline{\P}=\{\p_4,\p_6\}$, and the remaining maximal coding sets are $\overline{\S}=\mathcal{A}\setminus\S$=$\{(\p_3,\p_4)$, $(\p_4,\p_6)$, $(\p_5,\p_6)\}\}$. Since $\p_4$ has a diversity of 2 under $\overline{\S}$ due to $(\p_3,\p_4)$ and $(\p_4,\p_6)$, we branch $\S$ into two successors: $\S_1=\{(\p_1,\p_3)$,$(\p_2,\p_3,\p_5)$,$(\p_4,\p_5)\}$ and $\S_2=\{(\p_1,\p_3)$,$(\p_2,\p_3,\p_5)$,$(\p_4,\p_6)\}$. Since $\S_2$ contains all data packets and there are no other branching opportunities, the algorithm stops and outputs $\S_2$ as $\S_m$.
\end{enumerate}
\end{Example}


\begin{algorithm}[t]
\caption{Optimal S-IDNC solution search}\label{algorithm:branch}
\begin{algorithmic}[1]
\STATE \textbf{input:} the group of all maximal coding sets, $\mathcal{A}$;
\STATE \textbf{initialization:} a set $\mathcal{B}$ of solutions, $\mathcal{B}$ only contains an empty solution $\S=\emptyset$. A counter $u=1$;
\WHILE {no solution in $\mathcal{B}$ contains all data packets,}
\WHILE {there is a solution in $\mathcal{B}$ with size $u-1$,}
\STATE Denote this solution by $\S=\{\M_1,\cdots,\M_{u-1}\}$. Denote the data packets included in $\S$ by $\P=\bigcup_{i=1}^{u-1}\{\M_i\}$ and all data packets not included in $\S$ by $\overline{\P}=\P_K\setminus\P$. Also denote the maximal coding sets not included in $\S$ by $\overline{\S}=\mathcal{A}\setminus\S$;
\STATE Pick from $\overline{\P}$ the data packet $\p$ that has the smallest diversity $d$ under $\overline{\S}$. Denote the $d$ coding sets which contain $\p$ by $\M'_1,\cdots,\M'_d$;
\STATE Branch $\S$ into $d$ new solutions, $\S'_1,\cdots,\S'_d$. Then, add $\M'_1,\cdots,\M'_d$ to these solutions, respectively. The sizes of the new solutions are $u$;
\ENDWHILE
\STATE $u=u+1$;
\ENDWHILE
\STATE Output the solutions in $\mathcal{B}$ that contain all data packets.
\end{algorithmic}
\end{algorithm}

B-K algorithm and Algorithm 1 constitute our optimal S-IDNC coding algorithm. It produces all the valid $\S_m$. Among these solutions, we can choose the one that optimizes a secondary criteria, such as the one offering the smallest $D_\S$, or the largest $P_S$, calculated using \eqref{eq:p_s}.

\subsection{Hybrid S-IDNC Coding Algorithm}\label{sec:hybrid_algorithms}
Algorithm \ref{algorithm:branch} is memory demanding, because the number of candidate solutions grows exponentially with branching. Thus, we propose a heuristic alternative to it. The idea is to iteratively maximize the number of data packets included in $\S_m$. The algorithm is given in Algorithm 2.

\begin{algorithm}[t]
\caption{Hybrid S-IDNC solution search}\label{algorithm:hybrid}
\begin{algorithmic}[1]
\STATE \textbf{input:} the group of all maximal coding sets, $\mathcal{A}$;
\STATE \textbf{initialization:}  an empty solution $\S=\emptyset$, a counter $u=1$, packet set $\P=\P_K$;
\WHILE {$\S$ does not contain all data packets,}
\STATE find the coding set $\M$ in $\mathcal{A}$ that contains the largest number of data packets in $\P$;
\STATE Add $\M$ to $\S$ and remove data packets in $\M$ from $\P$;
\STATE $u=u+1$;
\ENDWHILE
\STATE Output the solution $\S$.
\end{algorithmic}
\end{algorithm}

B-K algorithm and Algorithm 2 constitute our hybrid S-IDNC coding algorithm. It produces only one S-IDNC solution, with no guarantee on the solution size. It is still computational expensive due to B-K algorithm. Thus, we develop a polynomial time heuristic S-IDNC coding algorithm next.

\subsection{Heuristic S-IDNC coding Algorithm}
Algorithm \ref{algorithm:clique} is a simple algorithm that heuristically finds the maximum (the largest maximal) clique of a graph. The intuition behind this algorithm is that, a vertex is very likely to be in the maximum clique if it is incident by the largest number of edges. Variations of this algorithm have been developed in the literature \cite{Rozner_Heuristic_clique,sorour:valaee:2010,sameh:valaee:globecom:2010}. But this algorithm has not been applied to finding a complete S-IDNC solution, and its computational complexity has not been identified yet.

\begin{algorithm}
\caption{Heuristic maximum clique search}\label{algorithm:clique}
\begin{algorithmic}[1]
\STATE \textbf{input}: graph $\G(\V,\E)$;
\STATE \textbf{initialization}: an empty vertex set $\V_{\textrm{keep}}$;
\WHILE {$\G$ is not empty}
\STATE add to  $\V_{\textrm{keep}}$ the vertex $\v$ which has the largest number of edges incident to it;
\STATE update $\G$ by deleting $\v$ and all the vertices not connected to $\v$ (\textit{These vertices can be ignored because they cannot be part of the target clique, which contains $\v$});
\ENDWHILE
\STATE vertices in $\V_{\textrm{keep}}$ are pair-wise connected, and no other vertices can be added to them. Hence, $\V_{\textrm{keep}}$ is a maximal clique.
\end{algorithmic}
\end{algorithm}

The computational complexity of Algorithm 3 is polynomial in the number of data packets $K$. The highest computational cost occurs when the input graph is complete, i.e., when all vertices are connected to each other.
In this case, only one vertex will be removed in each iteration. Thus, the number of remaining vertices in iteration-$i$ will be $K-i$, $\forall i\in[0,K-1]$. Then, to find the vertex with the largest number of incident edges, we need $K-i$ comparisons. The total computational cost is thus in the order of $\sum_{i=0}^{K-1}K-i=K(K-1)/2$. Hence, the computational complexity of Algorithm 3 is at most $\mathcal O(K^2)$.


We apply Algorithm 3 to iteratively find $\S_m$ in Algorithm \ref{algorithm:encoding}. In each iteration, we find a clique using Algorithm 3, maximize it by adding more vertices to it whenever possible, and then remove it from the S-IDNC graph. This will increase the diversities of the added vertices/packets. Below is an example:
\begin{Example}
Consider the graph $\G_s$ in \figref{fig:sg_example}. In the first two iterations, the algorithm will choose $\M_1=(\p_1,\p_2,\p_4)$ and $\M_2=(\p_3,\p_6)$, respectively. In the third iteration, $\V_\textrm{covered}=\{\p_1,\p_2,\p_3,\p_4,\p_6\}$ and the algorithm can only choose $\M_3=(\p_5)$. Among all the data packets in $\V_\textrm{covered}$, $\p_2$ can be added to $\M_3$. Thus $\M_3=\{\p_2,\p_5\}$. The algorithm then stops and outputs $\S_m=\{(\p_1,\p_2,\p_4),(\p_3,\p_6),(\p_2,\p_5)\}$.
\end{Example}
\begin{algorithm}
\caption{Heuristic S-IDNC solution search}\label{algorithm:encoding}
\begin{algorithmic}[1]
\STATE \textbf{input}: a graph $\G(\V,\E)$;
\STATE \textbf{initialization}: an empty vertex set $\V_{\textrm{covered}}$, a working graph $\G_w=\G$, and a counter $i=0$;
\WHILE {$\V_{\textrm{covered}}\neq\V$}
\STATE Find the maximum clique in $\G_w$ using Algorithm \ref{algorithm:clique}. Denote it by $\M_i$ ;
\STATE Find the vertices in $\V_{\textrm{covered}}$ which are connected to $\M_i$. Denote their set by $\V_i$ (\textit{They are the candidate vertices that could be added to $\M_i$}.);
\STATE Generate a subgraph of $\G$ whose vertex set is $\V_i$. Denoted this subgraph by $\G'_i(\V_i,\E_i)$;
\STATE Find the maximum clique in $\G'_i$ using Algorithm 1, denoted it by $\M_i'$ (\textit{All vertices in $\M_i'$ are connected to each other and thus can all be added to $\M_i$.});
\STATE Update $\V_{\textrm{covered}}$ by adding vertices in $\M_i$ into it;
\STATE Update $\G_w$ by removing $\M_i$ from it;
\STATE Update $\M_i$ as $\M_i=\M_i\cup\M_i'$ (\textit{The new clique is at least as large as the old one, and thus provides higher packet diversity});
\STATE $i=i+1$;
\ENDWHILE
\end{algorithmic}
\end{algorithm}

In conclusion, we proposed an optimal algorithm that finds $\S_m$, as well as its hybrid and heuristic alternatives. The output $\S_m$ is used as the S-IDNC solution for the semi-online transmission scheme. If fully-online transmission scheme is applied, the transmitted coding set $\M_f$ is chosen from $\S_m$.

\section{Simulations}
In this section, we present the simulated throughput and decoding delay performance of S-IDNC (abbreviated as S- in the figures) under different scenarios, including under full- and semi-online transmission schemes, and under the use of optimal, hybrid, and heuristic coding algorithms (abbreviated in the figures as Opt., Hybr., and Heur., respectively). The packet block size is $K=15$. The number of receivers $N$ is chosen between 5 and 40. The packet erasures are i.i.d. among the channels between the sender and the receivers probability with a probability of $P_e=0.2$.


We also compare S-IDNC with RLNC and G-IDNC. For RLNC, we assume a sufficiently large finite field, so that its throughput is almost surely optimal and serves as a benchmark. For G-IDNC, although its best performance is at least as good as S-IDNC (as we have explained in Section \ref{sec:sg_comp_erasure_free}), this advantage will not necessarily be reflected in our simulation results. This is because there has not been any optimal G-IDNC algorithm. Instead, we apply a heuristic algorithm (abbreviated as Heur. G- in the figures) proposed in \cite{sorour:valaee:2010}, which aims at minimizing the block completion time. This aim coincides with our optimization priorities for S-IDNC in Remark \ref{remark:U_prior}, namely, to minimize the block completion time first.

We conduct three sets of simulations. The first set compares the performance limits of the three techniques. The results are presented in \figref{fig:UD_min}. Here for RLNC, its absolute minimum block completion time is equal to the size of the largest Wants set of the receivers. This number cannot be further reduced by any means, because otherwise the most demanding receivers cannot decode all its wanted data packets. The second (resp. third) set of simulations compares the throughput and decoding delay performance under fully-online (resp. semi-online) transmission scheme. The results are presented in \figref{fig:UD_fully} (resp. \figref{fig:UD_semi}). We note that the performance of RLNC is the same under both schemes, because RLNC is feedback-free.

Our observations on S-IDNC are as follows:
\begin{itemize}
\item The absolute minimum block completion time of S-IDNC increases almost linearly with $N$. This result matches Corollary \ref{cor:U_linear_N};
\item The fully-online transmission scheme always provides better throughput and decoding delay performance than the semi-online one;
\item The optimal coding algorithm always provides better throughput performance than its hybrid and heuristic alternatives. This result verifies our choice of $\S_m$ for throughput improvement, because only the optimal coding  algorithm can always produce $\S_m$, which has $|\S_m|=U_s$;
\item However, the optimal coding  algorithm does not necessarily minimize the APDD. For example, in \figref{fig:UD_fully}(b), the hybrid algorithm provides smaller APDD than the optimal one under the fully-online transmission scheme;
\item The performance gap between the optimal and hybrid algorithms is always marginal, and is much smaller than their gap with the heuristic one. Hence, the hybrid algorithm strikes a good balance between performance and computational load.
\end{itemize}

A cross comparison of RLNC, S-, and G-IDNC shows that:
\begin{itemize}
\item The throughput of RLNC is always the best. The throughput of S-IDNC is very close to RLNC when the number of receivers is small. Their gap increases with $N$;
\item In general, the APDD of both S- and G-IDNC is better than RLNC. This advantage only vanishes when the block completion time of S- and G-IDNC becomes much larger than RLNC, which takes place when $N$ is much larger than $K$;
\item There is no clear winner between the performance of heuristic G-IDNC and optimal S-IDNC. We can expect that G-IDNC will outperform S-IDNC if its optimal coding algorithm is developed.
\end{itemize}

In summary, our simulations verified our theorems, propositions, and algorithms. They also demonstrated that, if we are concerned with both throughput and decoding delay performance, S-IDNC is a good alternative to RLNC when the number of receivers is not too large.

\begin{figure*}
\centering
\subfigure[Block completion time]{\includegraphics[width=0.45\linewidth]{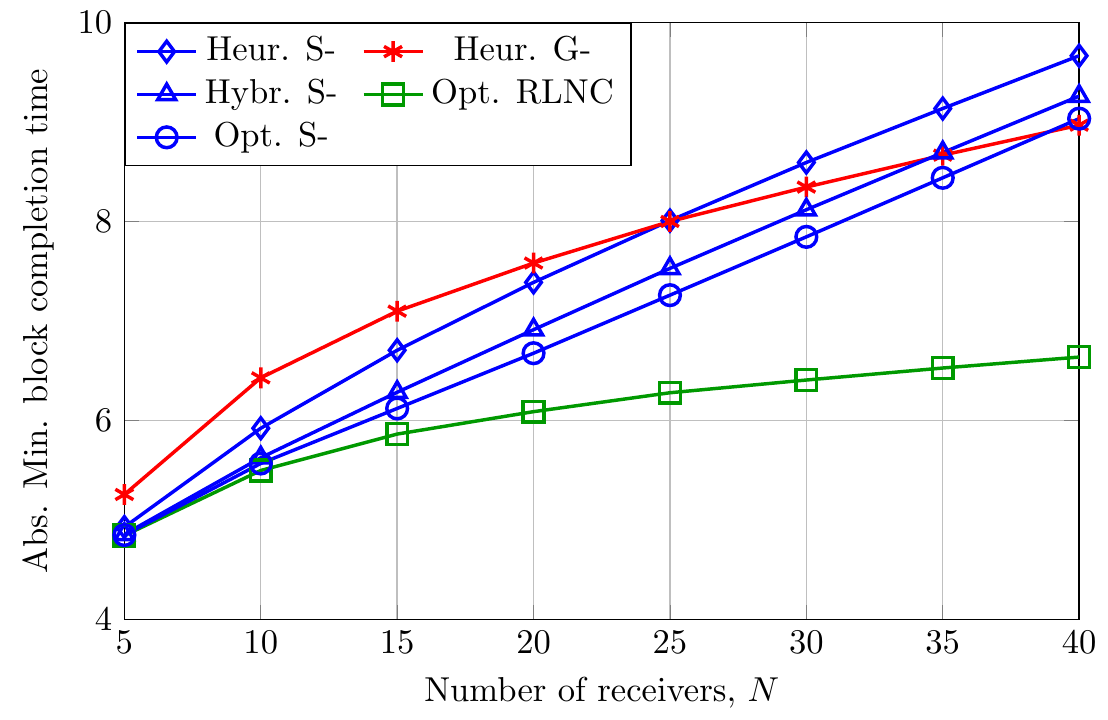}}
\subfigure[Average packet decoding delay]{\includegraphics[width=0.45\linewidth]{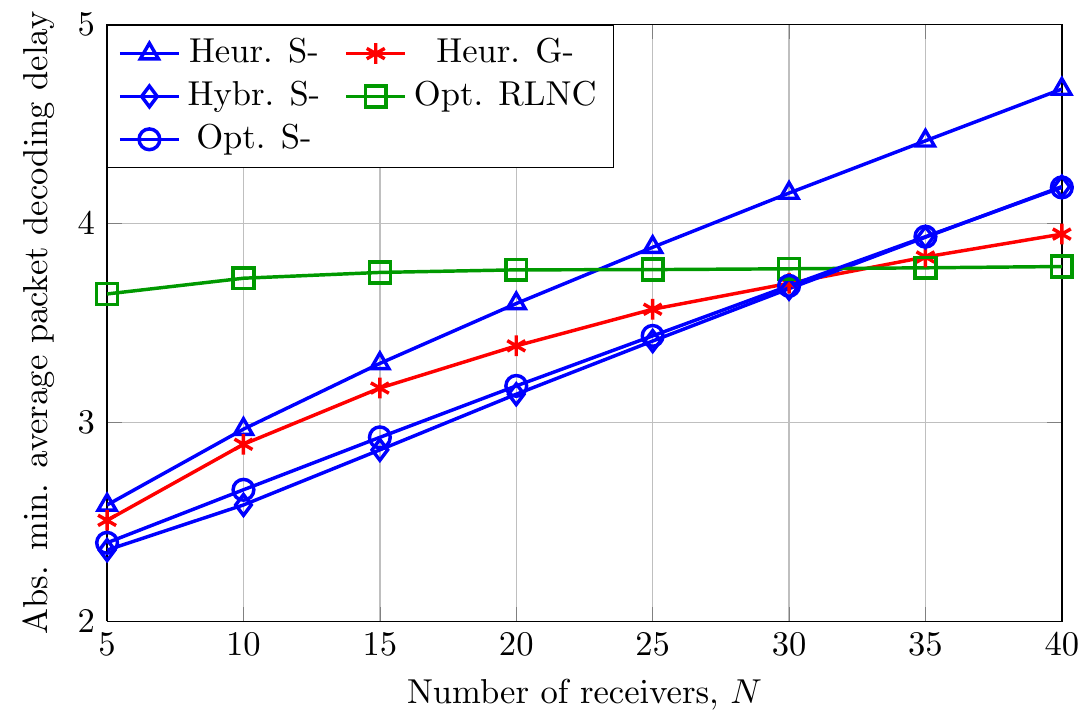}}
\caption{The throughput and decoding delay performance limits of S- and G-IDNC, as well as RLNC.}
\label{fig:UD_min}
\end{figure*}

\begin{figure*}
\centering
\subfigure[Block completion time]{\includegraphics[width=0.45\linewidth]{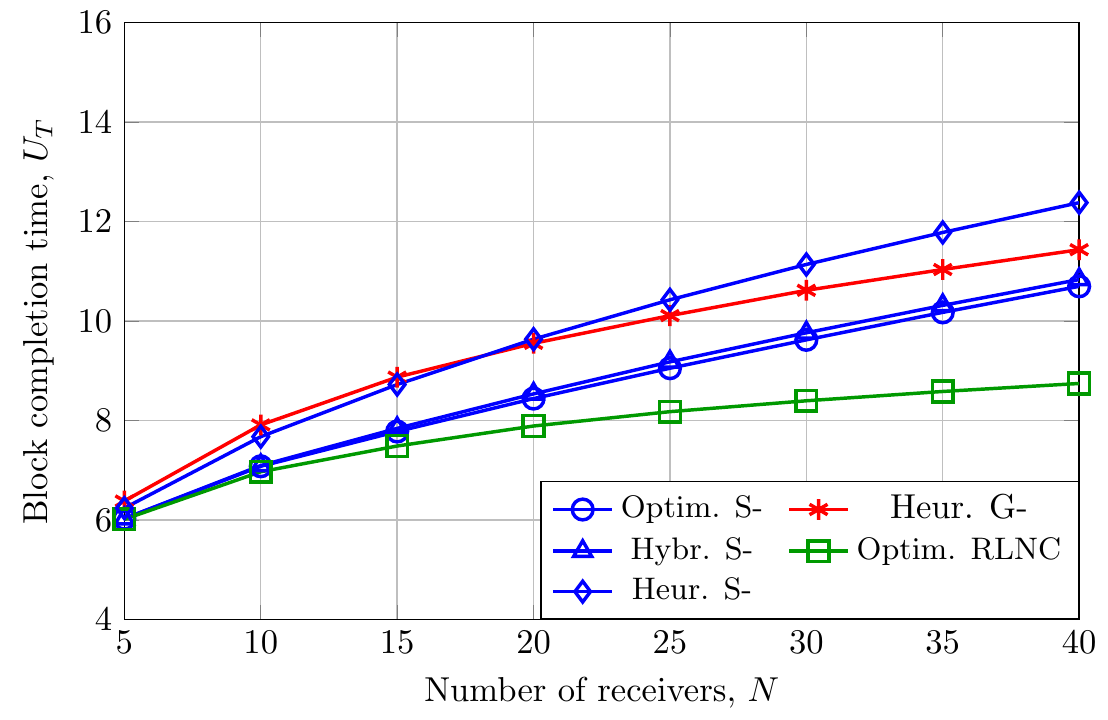}}
\subfigure[Average packet decoding delay]{\includegraphics[width=0.45\linewidth]{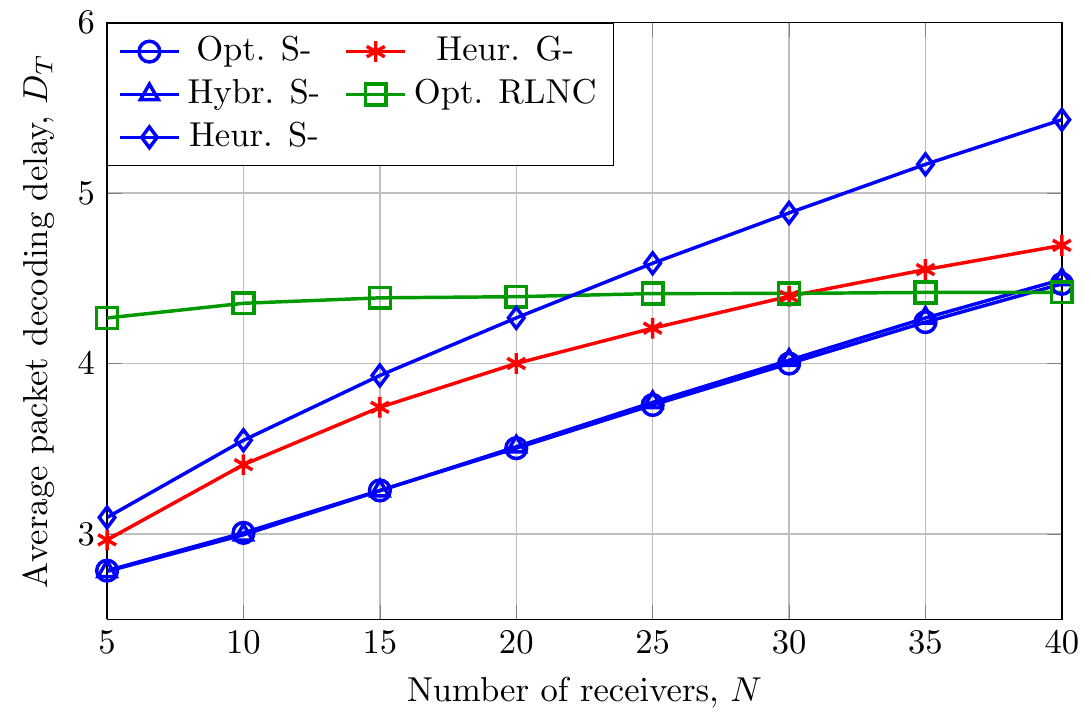}}
\caption{The throughput and decoding delay performance of fully-online transmission scheme when different coding algorithms are applied. It is compared with the performance of heuristic fully-online G-IDNC and RLNC.}
\label{fig:UD_fully}
\end{figure*}

\begin{figure*}
\centering
\subfigure[Block completion time]{\includegraphics[width=0.45\linewidth]{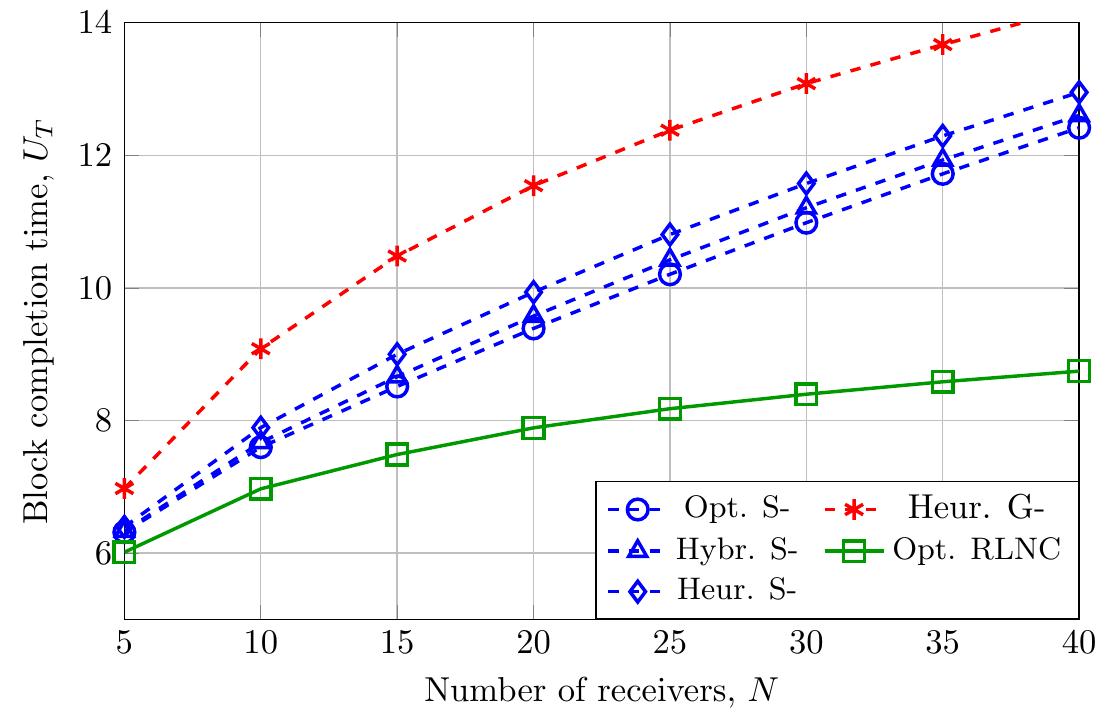}}
\subfigure[Average packet decoding delay]{\includegraphics[width=0.45\linewidth]{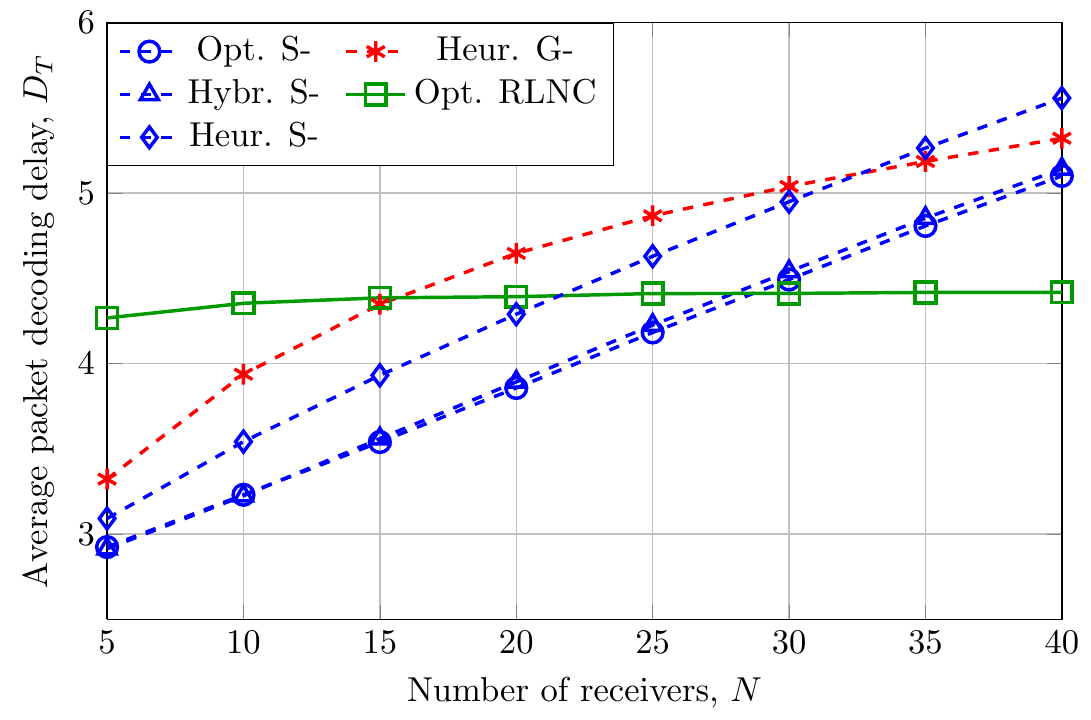}}
\caption{The throughput and decoding delay performance of semi-online transmission scheme when different coding algorithms are applied. It is compared with the performance of heuristic semi-online G-IDNC and RLNC.}
\label{fig:UD_semi}
\end{figure*}

\section{Conclusion}
In this paper, we studied the throughput and decoding delay performance of S-IDNC in broadcasting a block of data packets to wireless receivers under packet erasures. By using a random graph model, we showed that the throughput of S-IDNC decreases linearly with increasing number of receivers. By introducing the concept of perfect S-IDNC solution, we proved the NP-hardness of minimizing the average packet decoding delay. We also proposed two upper bounds on the throughput and decoding delay limits of S-IDNC.  We considered two transmission schemes that requires fully- and semi- online feedback frequencies, respectively. By applying stochastic shortest path method, we showed that it is intractable to make optimal coding decisions in the presence of random packet erasures. We then used heuristic objective functions to determine the preferred coded packet(s) to send. We then developed the optimal S-IDNC coding algorithm and its complexity-reduced heuristics. We also compared S-IDNC with G-IDNC by proving the equivalence between the chromatic number of the complementary S-IDNC and G-IDNC graphs. We used this equivalence to show that G-IDNC can outperform S-IDNC when there are not packet erasures, but this is not always true when there are packet erasures.

Our work provides news understandings of S-IDNC. It will facilitate the extension of S-IDNC to applications in other network settings, such as cooperative data exchange and distributed data storage. We are also interested in designing approximation and heuristic algorithms for decoding delay minimization.

\appendices
\renewcommand\thesection{Appendix \textnormal{section}}
\section{Proof of Theorem \ref{theo:equal_solutions}}

Theorem \ref{theo:equal_solutions} requires the proof of $\chi(\cG_s)=\chi(\cG_g)$. Since every S-IDNC solution is also a G-IDNC solution, but a G-IDNC solution is not necessarily an S-IDNC solution, we have $U_s\geqslant U_g$, and thus $\chi(\cG_s)\geqslant\chi(\cG_g)$. Hence, here we only need to prove that $\chi(\cG_s)\leqslant\chi(\cG_g)$.

We first introduce the concept of \emph{affiliated} S-IDNC graph $\G_{as}$ of a G-IDNC graph $\G_g$, which is construct as follows. Given $\G_g$ that involves $K$ data packets and $N$ receivers, we generate a graph $\G_{as}$ with $K$ vertices, each representing a data packet. We then connect $\v_i$ and $\v_j$ in $\G_{as}$ if for every pair of $\{m,n\}\in[1,N]$, $\v_{i,m}$ and $\v_{j,n}$ are connected upon their existence in $\G_g$. In other words, we claim that $\p_i$ and $\p_j$ do not conflict if every vertex that represents $\p_i$ in $\G_g$ is connected to every vertex that represents $\p_j$ in $\G_g$.

Given an SFM $\mA$, we can easily show that its S-IDNC graph $\G_s$ is the same as the affiliated S-IDNC graph $\G_{as}$ of its G-IDNC graph $\G_g$. Hence, our task becomes to prove that $\chi(\cG_{as})\leqslant\chi(\cG_g)$, where $\chi(\cG_{as})=U_s$. This statement is true if the following property is true:
\begin{property}\label{propt:update}
After removing any clique $\M_g$ from $\G_g$, the chromatic number of the affiliated S-IDNC graph $\G_{as}$ is reduced by at most one.
\end{property}
Since $\G_g$ is nonempty as long as $\G_{as}$ is nonempty, this property indicates that any clique partition solution of $\G_g$ must have a size of at least $\chi(\cG_{as})$, which will prove that $\chi(\cG_{as})\leqslant\chi(\cG_g)$. Property 6 can be proved through induction:
\begin{enumerate}
\item If $\M_g$ does not contain any conflicting data packets in $\G_{as}$, then $\chi(\cG_{as})$ is reduced by at most one;
\item If $\M_g$ contains one pair of conflicting data packets in $\G_{as}$, then $\chi(\cG_s)$ is reduced by at most one;
\item If $\M_g$ already contains $m$ pairs of conflicting data packets in $\G_{as}$, then modifying $\M_g$ to contain one more pair of conflicting data packets in $\G_{as}$ cannot further reduce $\chi(\cG_{as})$.
\end{enumerate}

The first statement is self-evident, because the set of data packets included in such $\M_g$ is a clique of $\G_{as}$. By removing it, $\chi(\cG_{as})$ can be reduced by at most one.

To prove the second statement, without loss of generality we assume that the pair of conflicting data packets is $(\p_1,\p_2)$. Then the set of data packets included in $\M_g$ takes a form of  $\{\M_s,\p_1,\p_2\}$, where $\M_s$ is the set of pair-wise non-conflicting data packets, and thus is a clique of $\G_{as}$. Since $\p_1$ conflicts with $\p_2$, there exists at least one pair of unconnected vertices in $\G_g$ that represent $\p_1$ and $\p_2$. This pair is not included in $\M_g$, and thus is kept after removing $\M_g$ from $\G_g$. Hence, in the updated affiliated S-IDNC graph $\G_{as}'$, $\v_1$ and $\v_2$ exist, and are unconnected. Let the chromatic number of $\G_{as}'$ be $U'$, then the minimum clique partition of $\G_{as}'$ takes a form of $\{\M_{1},\cdots,\M_{U'}\}$, which keeps $\p_1$ and $\p_2$ in different coding sets. Then, since $\M_s$ is a clique of $\G_{as}$, $\{\M_s,\M_1,\cdots,\M_{U'}\}$ is a partition of $\G_{as}$ with a size of $U'+1$. Thus, $U'\geqslant\chi(\cG_{as})-1$, implying that $\chi(\cG_s)$ is reduced by at most one after removing $\M_g$ from $\G_g$.

The proof of the third statement is similar to the second one, and thus is omitted here. According to the above three statements, no matter how many conflicting data packets are included in $\M_g$, after removing $\M_g$ from $\G_g$, the chromatic number of the affiliated S-IDNC graph $\G_{as}$ is reduced by at most one. Therefore, $\chi(\cG_g)\geqslant\chi(\cG_{as})$. Since $\G_{as}$ is the same as $\G_s$, we have $\chi(\cG_g)\geqslant\chi(\cG_{s})$ and Theorem 3 is proved.

\bibliographystyle{IEEEtran}

\end{document}